\documentclass[aps,superscriptaddress,twocolumn,twoside,floatfix,prl,a4paper]{revtex4-2}
\usepackage{amsmath,amssymb,amsthm}
\usepackage{amsmath,kbordermatrix,blkarray}
\usepackage[colorlinks=true,citecolor=blue,urlcolor=blue]{hyperref}
\usepackage[pdftex]{graphicx}
\usepackage{times,txfonts}
\usepackage{braket}
\usepackage{color}
\usepackage{mathtools}
\usepackage{natbib}
\usepackage{slashbox}

\allowdisplaybreaks

\newcommand{\tr}{\mathrm{Tr}}

\newcommand{\be}{\begin{equation}}
\newcommand{\ee}{\end{equation}}
\newcommand{\ba}{\begin{eqnarray}}
\newcommand{\ea}{\end{eqnarray}}

\newtheorem{observation}{Observation}
\newtheorem{definition}{Definition}
\newtheorem{thm}{Theorem}

\newtheorem{lem}{Lemma}

\begin{document}
	
\title{Asymmetric One-Sided Semi-Device-Independent Steerability of Quantum Discordant States}
\author{Chellasamy Jebarathinam}
  \email{jebarathinam@gmail.com}

 \affiliation{Department of Physics and Center for Quantum Information Science, National Cheng Kung University, Tainan 701, Taiwan}
 
 \author{Debarshi Das}
  \email{dasdebarshi90@gmail.com}
\affiliation{Department of Physics and Astronomy, University College London, Gower Street, WC1E 6BT London, England, United Kingdom}


\author{R. Srikanth}
\affiliation{Theoretical Sciences Division, Poornaprajna Institute of Scientific Research (PPISR), Bidalur post, Devanahalli, Bengaluru 562164, India}

\begin{abstract}
Superlocality and superunsteerability provide operational characterization of quantum correlations in certain local and unsteerable states respectively.
Such quantum correlated states
have a nonzero quantum discord.  A two-way nonzero quantum discord  is necessary for quantum correlations pointed out by superlocality.  On the other hand, in this work, we demonstrate that a two-way nonzero quantum discord is not necessary to demonstrate superunsteerability. To this end, we demonstrate superunsteerability for one-way quantum discordant states. This in turn implies the existence of one-way superunsteerability and also the presence of superunsteerability without superlocality. Superunsteerability for nonzero quantum discord states implies  the occurence of steerability in a one-sided semi-device-independent way. Just like one-way steerability occurs for certain Bell-local states in a one-sided device-independent way, our result shows that one-way steerability can also occur for certain nonsuperlocal states but  in a one-sided semi-device-independent way.
\end{abstract} 
	
	
	\maketitle
	
\textit{Introduction:-} Local quantum measurements on entangled states can be used to demonstrate quantum nonlocality, originating from an experimental situation proposed by Einstein, Podolsky and Rosen \cite{EPR35} and the Bohm-Aharonov version of it \cite{BA57}. Bell proposed a framework to distinguish quantum nonlocality from local realistic description of the measurement results  by introducing an inequality, which is satisfied by any local hidden variable model for the observed correlations between space-like separated observers \cite{Bel64}. Such an inequality is violated by certain quantum correlations and the phenomenon is referred as Bell nonlocality \cite{BCP+14}. There exists another form of quantum nonlocality as pointed out by Schr\"{o}dinger \cite{Sch35}. This form of quantum nonlocality is called quantum steering and its framework analogous to the Bell's framework was proposed by Wiseman, Jones and Doherty (WJD) \cite{WJD07}.  
Apart from being fundamental aspects of quantum theory,
both the forms of quantum nonlocality find applications for quantum technologies  
(see Sec IV in \cite{BCP+14} and Sec V in \cite{UCN+19} for applications of Bell nonlocality and quantum steering respectively).  In contrast to Bell nonlocality of quantum correlations, quantum steering is an asymmetric form of quantum correlations both from fundamental and their applications point of view.
Quantum steering can exist in only one-way, that is, certain entangled states have steerability from only one side \cite{BVQ+14} (also see Sec. III. D in \cite{UCN+19}) and quantum steering can only provide one-sided device-independent applications \cite{BCW+12} (also see Sec. V in \cite{UCN+19}).  

Quantification of quantum resources through appropriate quantifiers is an important aspect of quantum information science \cite{CG19}. Quantification of quantum correlations beyond entanglement, called quantum discord, was proposed in \cite{OZ01, HV01}. This kind of quantum correlation has also emerged as a quantum resource for applications in quantum information science \cite{DSC08,Pir14} (and also see Sec VI in \cite{MBC+12}).  
From a quantum foundational perspective \cite{Boh35}, quantum discord was proposed as  Bohr's notion of non-mechanical disturbance \cite{Wis13}. Certain distinguishing features of quantum discord to quantum entanglement have been characterized  such as no death for discord \cite{FCC+10} and quantum discord may increase under certain decoherence conditions  \cite{SKB11}. 

Simulation of certain local and unsteerable states  using finite shared randomness  has been shown to motivate the amount of shared randomness as resource \cite{BHQ+15,ZZ19}.   Superlocality \cite{DW15} and superunsteerability  \cite{DBD+18} have been recently formalized to demonstrate a quantum advantage in simulating certain local and unsteerable correlations, respectively, in  terms of local Hilbert space dimension over the minimal amount of shared randomness required to simulate them.  
Such quantum advantage has been invoked to provide operational characterization of quantum correlations in certain local and unsteerable states having a nonzero quantum discord \cite{JAS17, JDS+18,JDS+18,DJB+18}. Superlocality or superunsteerability has also been found to be useful for certifying quantum discord  in a measurement device-independent way \cite{JD23} and also as a resource for measurement device-independent quantum key distribution protocols \cite{GBS16}, quantum random access codes \cite{JDK+19} and  quantum random number generation \cite{JD23}.   

Studying the precise relationships among quantum discord, superunsteerability and superlocality could provide better understanding of quantum correlations as well as their role as a resource in quantum information processing. Superlocality is inequivalent to quantum discord \cite{JAS17}. This arises the question whether superunsteerability is inequivalent to superlocality or quantum discord. 
Also, as superunsteerability is an asymmetric concept, a natural question that arises is of whether superunsteerability can occur for nonzero quantum discord states with the one-wayness property, analogous to one-way quantum steering in the case of certain entangled states. In this work, we answer this question in the affirmative by demonstrating that quantum correlations in certain one-way quantum discordant states can be operationally captured by superunsteerability. For such states, superunsteerability cannot occur both the ways, because the state has zero quantum discord on one side. Thus, in this work, we demonstrate the existence of one-way superunsteerability.
This in turn implies that superunsteerability is inequivalent to superlocality.

\textit{WJD's form of Quantum Steering:-} Let us consider a steering scenario where two spatially separated parties, say Alice and Bob, share an unknown quantum system $\rho_{AB}\in \mathcal{B}(\mathcal{H}_A \otimes \mathcal{H}_B)$. Here $\mathcal{B}(\mathcal{H}_A \otimes \mathcal{H}_B)$ stands for the set of all bounded linear operators acting on the Hilbert space $\mathcal{H}_A \otimes \mathcal{H}_B$.  Alice 
performs a set of black-box measurements  and the Hilbert-space dimension of Bob's 
subsystem is known. Such a scenario is called one-sided device-independent since
Alice's measurement operators $\{M_{a|x}\}_{a,x}$, which are  positive operator-valued measures (POVM), are unknown. The steering scenario 
is completely characterized by the set of unnormalized conditional states on Bob's
side $\{\sigma_{a|x}\}_{a,x}$, which is called an unnormalized assemblage. Each element
in the unnormalized assemblage is given by $\sigma_{a|x}=p(a|x)\rho_{a|x}$,  where $p(a|x)$ is the conditional probability of getting the outcome $a$ when Alice performs the measurement $x$;
$\rho_{a|x}$ is the normalized conditional state on Bob's side.
Quantum theory predicts that all valid assemblages should satisfy the following criteria:
\begin{equation}
\sigma_{a|x}=\tr_A ( M_{a|x} \otimes \openone \rho_{AB}) \hspace{0.5cm} \forall \sigma_{a|x} \in \{\sigma_{a|x}\}_{a,x}.
\end{equation}

\begin{definition}
In the above scenario, Alice demonstrates WJD's form of steerability to Bob \cite{WJD07}  if the assemblage does not have a local hidden state (LHS) model, i.e., if for all $a$, $x$, there is no decomposition of $\sigma_{a|x}$ in the form,
\begin{equation}
\sigma_{a|x}=\sum_\lambda p(\lambda) p(a|x,\lambda) \rho_\lambda,
\end{equation}
where $\lambda$ denotes classical random variable which occurs with probability 
$p(\lambda)$; $\rho_{\lambda}$ are called local hidden states which satisfy $\rho_\lambda\ge0$ and
$\tr\rho_\lambda=1$.
\end{definition}

We can further define the detection of above phenomenon from the nosignaling (NS) boxes (see Appendix  A) as follows.
\begin{definition}
Suppose Bob performs a set of projective measurements $\{\Pi_{b|y}\}_{b,y}$ on $\{\sigma_{a|x}\}_{a,x}$. Then the scenario is characterized by the set of measurement correlations, or box
between Alice and Bob $P(ab|xy)$:=$\{p(ab|xy)\}_{a,x,b,y}$, where $p(ab|xy)$ = $\tr ( \Pi_{b|y} \sigma_{a|x} )$. The box $P(ab|xy)$ detects WJD's steerability from Alice to Bob, 
iff it does not have a decomposition as follows \cite{WJD07, JWD07}: 
\begin{equation}
p(ab|xy)= \sum_\lambda p(\lambda) p(a|x,\lambda) p(b|y, \rho_\lambda) \hspace{0.3cm} \forall a,x,b,y; \label{LHV-LHS}
\end{equation}
where, $\sum_{\lambda} p(\lambda) = 1$, $p(a|x, \lambda)$ denotes an arbitrary probability distribution arising from local hidden variable (LHV) $\lambda$ ($\lambda$ occurs with probability $p(\lambda)$) and $p(b|y, \rho_{\lambda}) $ denotes the quantum probability of outcome $b$ when measurement $y$ is performed on local hidden state (LHS) $\rho_{\lambda}$.
\end{definition}

\textit{Quantum discord:-} In the following we present the definition of quantum discord  \cite{HV01, OZ01} from Alice to Bob $D^{\rightarrow}(\rho_{AB})$.
Quantum discord is defined as
\begin{align}\label{QDdef}
D^{\rightarrow}(\rho_{AB})= I(\rho_{AB}) - C^{\rightarrow}(\rho_{AB})
            = S(\rho_{B|A}) - \tilde{S}(\rho_{B|A}).
\end{align}            
Here $I(\rho_{AB})= 
 S(\rho_B) - \tilde{S}(\rho_{B|A})$,
is the quantum mutual information, and can be interpreted as the total correlations in \(\rho_{AB}\). 
\(S(\sigma) = -\mbox{tr}(\sigma \log_2 \sigma)\) is the von Neumann entropy of a density matrix \(\sigma\). 
 \(\tilde{S}(\rho_{B|A}) = S(\rho_{AB}) - S(\rho_A)\) is the ``unmeasured'' quantum conditional entropy \cite{qmi} (see also \cite{Cerf, SN96,GROIS}). 
On the other hand,  \(C^{\rightarrow}(\rho_{AB}) = S(\rho_B) - S(\rho_{B|A})\) can be interpreted as the classical correlations in \(\rho_{AB}\),
where the quantum conditional entropy is defined as  
\(S(\rho_{B|A}) = \min_{\{M_i^A\}} \sum_i p_i S(\rho_{B|i})\),
with the minimization being over all 
POVMs, \(\{M^A_i\}\),  performed on subsystem \(A\).
Here \(p_i = \mbox{tr}_{AB}( M^A_i   \otimes \mathbb{I}_B \rho_{AB} M^A_i   \otimes \mathbb{I}_B )\) is the probability of obtaining the outcome \(i\), and 
the corresponding post-measurement state  
for the subsystem \(B\) is \(\rho_{B|i} = \frac{1}{p_i} \mbox{tr}_A( M^A_i \otimes  \mathbb{I}_B  \rho M^A_i \otimes  \mathbb{I}_B )\).

Quantum discord $D^{\rightarrow}(\rho_{AB})$ as defined above captures quantum correlation in the state  from Alice to Bob. Similarly, quantum discord  from Bob to Alice $D^{\leftarrow}(\rho_{AB})$ can be defined.  $D^{\rightarrow}(\rho_{AB})$  vanishes for a given $\rho_{AB}$ if and only if it is a classical-quantum state of the form,
\begin{align}
\rho_{\texttt{CQ}}=\sum_i \ket{i} \bra{i}_A \otimes \rho^{(i)}_B,
\end{align}
where $\{\ket{i}\}$ forms an orthonormal basis on Alice's Hilbert space and $\rho^{(i)}_B$ are any quantum states on Bob's Hilbert space. On the other hand, $D^{\leftarrow}(\rho_{AB})$  vanishes for a given $\rho_{AB}$ if and only if it is a quantum-classical state of the form,
\begin{align}
\rho_{\texttt{QC}}=\sum_i   \rho^{(i)}_A  \otimes \ket{i} \bra{i}_B,
\end{align}
where now $\{\ket{i}\}$ forms an orthonormal basis on Bob's Hilbert space and $\rho^{(i)}_A$ are any quantum states on Alice's Hilbert space.

\textit{Super-unsteerability:-} We are going to present the formal definition of the notion \textit{``super-unsteerability"} for boxes having a LHV-LHS model \cite{DBD+18}. Before that we present the definition of \textit{``superlocality"} for local correlations \cite{DW15,JAS17}. Consider a Bell scenario, where both parties perform black box measurements. In this scenario, superlocality is defined as follows:
\begin{definition}
Suppose we have a quantum state in $\mathbb{C}^{d_A}\otimes\mathbb{C}^{d_B}$
and measurements which produce a local bipartite box $P(ab|xy)$ := $\{ p(ab|xy) \}_{a,x,b,y}$.
Then, superlocality holds iff there is no decomposition of the box in the form,
\begin{equation}
p(ab|xy)=\sum^{d_\lambda-1}_{\lambda=0} p(\lambda) p(a|x, \lambda) p(b|y, \lambda) \hspace{0.3cm} \forall a,x,b,y,
\end{equation}
with dimension of the shared randomness/hidden variable $d_\lambda\le$ min($d_A$, $d_B$).  Here $\sum_{\lambda} p(\lambda) = 1$, $p(a|x, \lambda)$ and $p(b|y, \lambda)$ denotes arbitrary probability distributions arising from LHV $\lambda$ ($\lambda$ occurs with probability $p(\lambda)$).
\end{definition}
In Ref. \cite{JAS17}, an example of superlocality has been demonstrated with the noisy CHSH local box given by 
\begin{equation} 
P(ab|xy)       =       \frac{2+(-1)^{a\oplus       b\oplus
    xy}\sqrt{2} V}{8}, \label{chshfam}
\end{equation}
with $ 0<  V \le 1/\sqrt{2}$.
Such local correlations can be produced by a two-qubit pure  entangled state or a two-qubit Werner state having entanglement or a nonzero quantum discord for appropriate local noncommuting measurements, on the other hand, it cannot be reproduced by a LHV model with $d_\lambda=2$ as shown in Ref. \cite{JAS17}. 

Now, consider a different scenario where one of the parties (say, Alice) performs black box measurements and another party (say, Bob) performs quantum measurements. In this steering scenario, the notion of \textit{``super-unsteerability"} has been defined as follows:
\begin{definition}
Suppose we have a quantum state in $\mathbb{C}^{d_A}\otimes\mathbb{C}^{d_B}$
and measurements which produce a unsteerable bipartite  box $P(ab|xy)$ := $\{ p(ab|xy) \}_{a,x,b,y}$.
Then, super-unsteerability holds iff there is no decomposition of the box in the form,
\begin{equation}\label{LHSd}
p(ab|xy)=\sum^{d_\lambda-1}_{\lambda=0} p(\lambda) p(a|x, \lambda) p(b|y, \rho_{\lambda}) \hspace{0.3cm} \forall a,x,b,y,
\end{equation}
with dimension of the shared randomness/hidden variable $d_\lambda\le d_A$.  Here $\sum_{\lambda} p(\lambda) = 1$, $p(a|x, \lambda)$ denotes an arbitrary probability distribution arising from LHV $\lambda$ ($\lambda$ occurs with probability $p(\lambda)$) and $p(b|y, \rho_{\lambda}) $ denotes the quantum probability of outcome $b$ when measurement $y$ is performed on LHS $\rho_{\lambda}$ in $\mathbb{C}^{d_B}$.
\end{definition}
In Ref. \cite{DBD+18}, two examples of superunsteerability have been demonstrated. In one of these examples, the unsteerable white noise-BB84 family given by  
\begin{equation}
\label{bb84}
P(ab|xy) = \frac{1 + (-1)^{a \oplus b \oplus x.y} \delta_{x,y} V }{4},
\end{equation}
with $0 < V \leq 1/\sqrt{2}$,
can be produced by a two-qubit Werner state having entanglement or a nonzero quantum discord for appropriate local noncommuting measurements. On the other hand, it cannot be simulated by a LHV-LHS model with $d_\lambda=2$ as shown in Ref. \cite{DBD+18}. These two examples demonstrate superunsteerability both the ways as they are symmetrical with respect to interchanging Alice and Bob. In the following, we demonstrate an example of superunsteerability asymmetrically. 

\textit{One-way Superunsteerability:-} Consider that the two spatially separated parties (say, Alice and Bob) share the following separable two-qubit state,
\begin{equation}
\label{state2n}
\rho = \frac{1}{2} \Big( |00\rangle \langle 00| + |+1 \rangle \langle +1| \Big) ,
\end{equation}
where, $|0\rangle$ and $|1\rangle$ are the eigenstates of the operator $\sigma_z$ corresponding to the eigenvalue $+1$ and $-1$ respectively; $|+\rangle$ is the eigenstate of the operator $\sigma_x$ corresponding to the eigenvalue $+1$. The above state has quantum discord  $D^{\rightarrow}(\rho_{AB})>0$
and $D^{\leftarrow}(\rho_{AB})=0$ (see Appendix B for the definition of quantum discord) since it  is not a classical-quantum state but a quantum-classical state \cite{HV01,OZ01}. If Alice performs the projective measurements of observables corresponding to the operators $A_0 = \sigma_x$ and $A_1 = \sigma_z$, and Bob performs projective measurements of observables corresponding to the operators $B_0 =  \sigma_x$ and $B_1 = \sigma_z$, then the following correlation is produced from the above quantum-classical state,

 \begin{equation}      
 \label{corr1}
 P(ab|xy) =\begin{tabular}{c|cccc}
 \backslashbox{(x,y)}{(a,b)} & (0,0) & (0,1) & (1,0) & (1,1)\\\hline\\[0.05cm]
(0,0) & $\dfrac{1}{2}$ & $\dfrac{1}{4}$ & $0$ & $\dfrac{1}{4}$ \\[0.5cm]
(0,1) & $\dfrac{3}{8}$ & $\dfrac{3}{8}$ & $\dfrac{1}{8}$ & $\dfrac{1}{8}$ \\[0.5cm]
(1,0) & $\dfrac{1}{4}$ & $\dfrac{1}{2}$ & $\dfrac{1}{4}$ & $0$ \\[0.5cm]
(1,1) & $\dfrac{3}{8}$ & $\dfrac{3}{8}$ & $\dfrac{1}{8}$ & $\dfrac{1}{8}$ \\
\end{tabular}
\end{equation}
Here $x$, $y$ denote the input variables on Alice's and Bob's sides respectively; and $a$, $b$ denote the outputs on Alice's and Bob's sides respectively.  In the following, we demonstrate that the box (\ref{corr1}) detects super-unsteerability 
of the quantum-classical state (\ref{state2n}).

Let us now proceed to analyse simulating the correlation given by Eq.(\ref{corr1}) with LHV at one side and LHS at another side. Before proceeding, let us define the following (for details, see Appendix A) which will be used throughout the article,
\begin{equation}
P_D^{\alpha\beta}(a|x)=\left\{
\begin{array}{lr}
1, & a=\alpha x\oplus \beta\\
0 , & \text{otherwise}\\
\end{array}
\right. , \hspace{0.1cm} P_D^{\gamma\epsilon}(b|y)=\left\{
\begin{array}{lr}
1, & b=\gamma x\oplus \epsilon\\
0 , & \text{otherwise}.\\
\end{array}
\right. \nonumber 
\end{equation} 

The  correlation given by Eq.(\ref{corr1}) has the following LHV-LHS model:
\begin{equation}
P(ab|xy)= \sum_{\lambda=0}^{2} p(\lambda) P(a|x, \lambda) P(b|y, \rho_{\lambda}),
\label{eee}
\end{equation} 
where $p(0)$ = $\frac{1}{2}$, $p(1)$ = $p(2)$ = $\frac{1}{4}$; $P(a|x,0)$ = $P_D^{00}(a|x)$, $P(a|x,1)$ = $P_D^{10}(a|x)$, $P(a|x,2)$ = $P_D^{11}(a|x)$
 and
\begin{align}       
P(b|y,\rho_0)  
&=\! \begin{tabular}{c|cc}
 \backslashbox{(y)}{(b)} & (0) & (1) \\\hline
(0) & $\frac{1}{2}$ & $\frac{1}{2}$  \\
(1) & $\frac{1}{2}$ & $\frac{1}{2}$  \\
\end{tabular}  \nonumber \\
&=\!  \langle \psi^{'}_0 | \{\Pi_{b|y}\}_{b,y} | \psi^{'}_0 \rangle, \nonumber
\end{align}
\begin{align}      
P(b|y,\rho_1)  &=\! \begin{tabular}{c|cc}
 \backslashbox{(y)}{(b)} & (0) & (1) \\\hline
(0) & $1$ & $0$  \\
(1) & $\frac{1}{2}$ & $\frac{1}{2}$  \\
\end{tabular} \nonumber\\
&=\!  \langle \psi^{'}_1 | \{\Pi_{b|y}\}_{b,y} | \psi^{'}_1 \rangle, \nonumber
\end{align}
\begin{align}      
P(b|y,\rho_2) &=\! \begin{tabular}{c|cc}
 \backslashbox{(y)}{(b)} & (0) & (1) \\\hline
(0) & $0$ & $1$  \\
(1) & $\frac{1}{2}$  & $\frac{1}{2}$  \\
\end{tabular} \nonumber\\
&=\!  \langle \psi^{'}_2 | \{\Pi_{b|y}\}_{b,y} | \psi^{'}_2 \rangle, \nonumber
\end{align} 
 where $\{\Pi_{b|y}\}_{b,y}$ corresponds to two arbitrary projective mutually unbiased measurements in the Hilbert space $\mathcal{C}^2$ corresponding to the operators $B_0 = |\uparrow_0 \rangle \langle \uparrow_0|$ $-$ $|\downarrow_0 \rangle \langle \downarrow_0|$ and $B_1 =  |\uparrow_1 \rangle \langle \uparrow_1|$ $-$ $|\downarrow_1 \rangle \langle \downarrow_1|$;  here $\{ |\uparrow_0 \rangle$, $|\downarrow_0 \rangle \}$ is an arbitrary orthonormal basis in the Hilbert space $\mathcal{C}^2$ and the orthonormal basis $\{ |\uparrow_1 \rangle, |\downarrow_1 \rangle \}$ in the Hilbert space $\mathcal{C}^2$ is such that aforementioned two measurements define two arbitrary projective mutually unbiased measurements in the Hilbert space $\mathcal{C}^2$. The $|\psi^{'}_{\lambda}\rangle$s that produce
$p(b|y,\rho_{\lambda})$s given above are given by  
$|\psi^{'}_0 \rangle = \frac{1}{\sqrt{2}} |\uparrow_0 \rangle + i  \frac{1}{\sqrt{2}} |\downarrow_0 \rangle,$
$
|\psi^{'}_1 \rangle =  |\uparrow_0 \rangle,
$
and
$
|\psi^{'}_2 \rangle = |\downarrow_0 \rangle,
$
which are all valid states in the Hilbert space $\mathcal{C}^2$.\\

Hence, the LHV-LHS decomposition of  the correlation given by Eq.(\ref{corr1}) can be realized with hidden variable having dimension $3$ (with two arbitrary projective mutually unbiased measurements at trusted party). This is also the minimal hidden variable dimension needed to simulate the correlation as 
we show in the following lemma.
\begin{lem}
The LHV-LHS decomposition of  the correlation given by Eq.(\ref{corr1}) cannot be realized with hidden variable having dimension $2$.
\end{lem}
The proof of this lemma is given in Appendix  B.
We now show the following result. 
\begin{thm}
The correlation given by Eq.(\ref{corr1}) demonstrates super-unsteerablity from Alice to Bob while having no super-unsteerablity from Bob to Alice. 
\end{thm}
\begin{proof}
We have shown that the unsteerable correlation given by Eq.(\ref{corr1}) can have a LHV-LHS model
with the minimum dimension of the hidden variable being $3$. On the other hand, we have seen that the unsteerable correlation given by Eq.(\ref{corr1}) can be simulated by using $2 \otimes 2$ quantum system (\ref{state2n}). This is an instance of super-unsteerability since the minimum dimension of shared randomness needed to simulate the LHV-LHS model of the correlation (\ref{corr1}) is greater than the local Hilbert space dimension of the shared quantum system (reproducing the given unsteerable correlation) at the untrusted party's side (who steers the other party, in the present case Bob).

On the other hand, the box (\ref{corr1}) does not have superunsteerability from Bob to Alice. This is because it arises from a unsteerable state having zero discord from Bob to Alice, whereas non-zero discord from Bob to Alice is necessary for producing superunsteerable (from Bob to Alice) correlation \cite{DBD+18}. It can also be checked by providing a LHS-LHV model with $d_\lambda=2$ as follows. From the decomposition of the box (\ref{corr1}) in terms of the local deterministic boxes, we obtain the following  LHS-LHV model:
\begin{equation}
P(ab|xy)= \sum_{\lambda=0}^{1} p(\lambda)  P(a|x, \rho_{\lambda}) P(b|y, \lambda),
\label{eeee}
\end{equation}
where $p(0)$ = $\frac{1}{2}$, $p(1)$ = $\frac{1}{2}$;\\
\begin{align}      
P(a|x,\rho_0)   &=\! \begin{tabular}{c|cc}
 \backslashbox{(x)}{(a)} & (0) & (1) \\\hline
(0) & $1$ & $0$  \\
(1) & $\frac{1}{2}$ & $\frac{1}{2}$  \\
\end{tabular} \nonumber\\
&=\!  \langle \psi^{'}_0 | \{\Pi_{a|x}\}_{a,x} | \psi^{'}_0 \rangle, \nonumber 
\end{align} 
\begin{align}
P(a|x,\rho_1)  &=\! \begin{tabular}{c|cc}
 \backslashbox{(x)}{(a)} & (0) & (1) \\\hline
(0) & $0$ & $1$  \\
(1) & $\frac{1}{2}$  & $\frac{1}{2}$  
\end{tabular} \nonumber\\
&=\!  \langle \psi^{'}_1 | \{\Pi_{a|x}\}_{a,x} | \psi^{'}_1 \rangle, \nonumber 
\end{align} 
 where $\{\Pi_{a|x}\}_{a,x}$ corresponds to two arbitrary projective mutually unbiased measurements in the Hilbert space $\mathcal{C}^2$ corresponding to the operators $A_0 = |\uparrow_0 \rangle \langle \uparrow_0|$ $-$ $|\downarrow_0 \rangle \langle \downarrow_0|$ and $A_1 =  |\uparrow_1 \rangle \langle \uparrow_1|$ $-$ $|\downarrow_1 \rangle \langle \downarrow_1|$; here $\{ |\uparrow_0 \rangle$, $|\downarrow_0 \rangle \}$  and  $\{ |\uparrow_1 \rangle, |\downarrow_1 \rangle \}$ define two arbitrary projective mutually unbiased measurements in the Hilbert space $\mathcal{C}^2$. The $|\psi^{'}_{\lambda}\rangle$s that produce
$p(b|y,\rho_{\lambda})$s given above are given by  
$
|\psi^{'}_0 \rangle =  |\uparrow_0 \rangle
$
and
$
|\psi^{'}_1 \rangle = |\downarrow_0 \rangle,
$
which are all valid states in the Hilbert space $\mathcal{C}^2$;\\
 and
$P(b|y,0)$ = $(P_D^{00}(b|y)+P_D^{10}(b|y))/2$ and $P(b|y,1)$ = $(P_D^{00}(b|y)+P_D^{11}(b|y))/2$. 
\end{proof}

\begin{figure}[t!]
\begin{center}
\includegraphics[width=7.5cm]{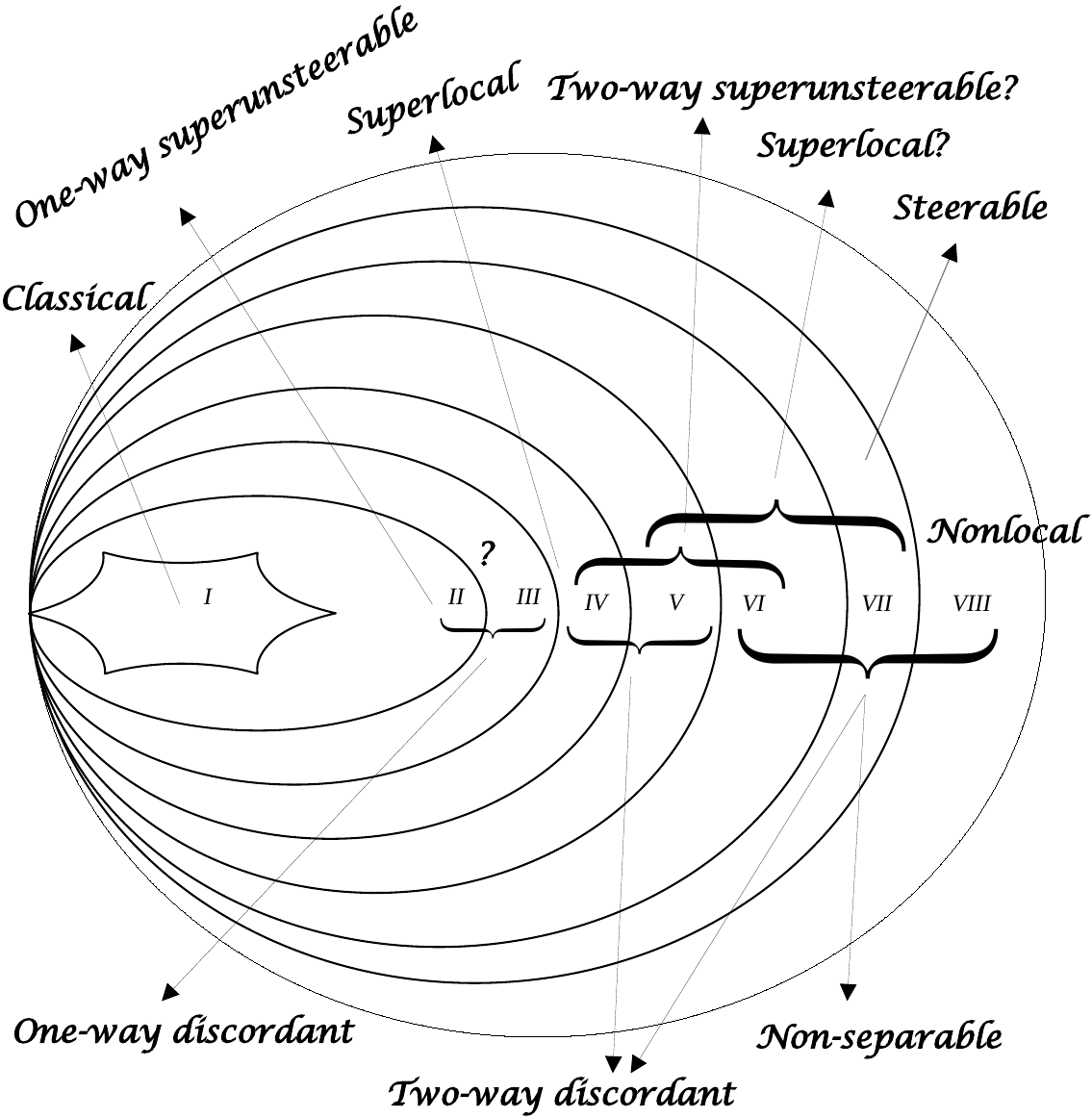}
\end{center}
\caption{Hierarchy of correlations in bipartite quantum states. 
The regions $I$, $II$, $III$, $IV$ and $V$ represent the convex subset of correlations in separable states. The complement
set, i.e., regions $VI$, $VII$ and $VIII$, represent non-seperable correlations in  entangled states. 
Within the non-separable correlations, the regions $VII$ and $VIII$ represent one-sided device-independent steerable and Bell nonlocal correlations, respectively, while the region $VI$ represents the non-separable correlations which are neither steerable nor Bell-nonlocal. On the other side, within the correlations in separable states,
correlations in one-way  discordant states
and two-way  discordant states are depicted in regions $II$ and $III$ and regions $IV$ and $V$, respectively, while the region $I$ represents the classical correlations that do not require a discordant state to be produced. Within the discordant separable correlations, the regions $II$ and $IV$ represent the 
one-way superunsteerable  and superlocal correlations that can be produced from one-way discordant and two-way discordant states, respectively. Whether one-way superunsteerable correlations form a strict subset of the correlations in one-way discordant states is an open question. Further, whether all correlations in regions $IV$, $V$ and $VI$ are two-way superunsteerable and whether all correlations in regions $V$, $VI$ and $VII$ are superlocal remain to be explored. Note that the noisy CHSH local box (\ref{chshfam}) exhibiting superlocality belongs to the regions $V$, $VI$ and $VII$, on the other hand,   the unsteerable white noise-BB84 family  (\ref{bb84}) exhibiting two-way superunsteerability and superlocality belongs to the regions $IV$, $V$ and $VI$.
\label{Fig:Hirarchy}}
\end{figure}

Note that Eq. (\ref{eeee}) can be seen as a LHV-LHV decomposition for the correlation with hidden variable dimension $d_\lambda=2$. Hence, the box (\ref{corr1}) is not superlocal. In a similar way, it can be easily checked that any correlation, which is not superunsteerable at least in one direction, is not superlocal as well.
Though the one-way discordant state (\ref{state2n}) does not have quantum correlation pointed out by superlocality,  the above result implies that it still have quantum correlation pointed out by superunsteerability asymmetrically. As an implication of this result,  in Fig. \ref{Fig:Hirarchy},  we depict a new hierarchy of quantum correlations in nonzero quantum discord states.

\textit{A weak form of  quantum  steering:-} Here we address the connection of superunsteerability to a weak form of  steering and its distinction from WJD's form of steerability at the level of certification of steerabillity. To this end, we note that in the given steering scenario, WJD's form of quantum steering implies the presence of steerability in a one-sided device-independent way, i.e., without making any assumption about the device used by the steering side  (i.e., untrusted side). On the other hand,  in the following, we define  another inequivalent form of  steerability to imply the presence of steerability in a one-sided semi-device-independent way, i.e., by assuming only the dimension of the steering side.
\begin{definition}
In the steering scenario as described before for defining WJD's form of steering, Alice demonstrates  steering to Bob in a one-sided semi-device-independent way if the assemblage does not have a LHS model with a restricted hidden variable dimension, i.e., if for all $a$, $x$, there is no decomposition of $\sigma_{a|x}$ in the form,
\begin{equation}\label{LHSdl}
\sigma_{a|x}=\sum^{d_\lambda-1}_{\lambda=0} p(\lambda) p(a|x,\lambda) \rho_\lambda,
\end{equation}
with $d_\lambda \le d_A$.
\end{definition}

Note that for a given assemblage which has a LHS model with the hidden variable dimension $d_\lambda$ as in Eq. (\ref{LHSdl}), there exists a suitable choice of POVMs on Bob's side, $\{M_{b|y}\}$, such it produces a LHV-LHS correlation $P(ab|xy)$ as in Eq. (\ref{LHSd}) with $p(b|y,\rho_\lambda)=\tr( M_{b|y} \rho_\lambda )$ and the same $d_\lambda$.
This implies that, in the context of the above definition of   steering, we have the following observation for the detection of it in a steering scenario where Bob performs particular measurements. 
\begin{observation}
A bipartite box detects steerability in a one-sided semi-device-independent way if and only if there is no decomposition of the box in the form given by Eq. (\ref{LHSd})  with $d_\lambda\le d_A$.
\end{observation}

From the above discussions it is clear that the correlation given by Eq. (\ref{corr1}) detects one-way steerability  in a one-sided semi-device-independent way. We have thus identified a new nonconvex subset of correlations in one-way discordant states, i.e., one-way superunsteerable correlations which do not exhibit superlocality as in Fig. \ref{Fig:Hirarchy}, but exhibit one-sided semi-device-independent steerability asymmetrically. There are superlocal correlations which exhibit two-way superunsteerability as in the example given by Eq. (\ref{bb84}) \cite{DBD+18},
and hence, they exhibit two-way steerability in a one-sided semi-device-independent way.

\textit{Conclusions:-} In this work, we have demonstrated the existence of a new asymmetric nature of quantum correlations  
in quantum discordant states. This asymmetric nature of quantumness arises due to one-way superunsteerability. For quantumness pointed out by superlocality, a two-way nonzero  quantum discord  is necessary. Whereas, superunsteerability being an asymmetric concept, to reveal quantumness pointed out by superunsteerability, a one-way nonzero discord suffices as we demonstrated in this work. This result helps us to obtain a precise relationship among quantum discord, superunsteerability and superlocality. We hope that this relationship would stimulate to investigate 
superunsteerability as a distinct resource than superlocality for quantum 
information processing.

Superunsteerability indicates a weak form of 
steerability, i.e., if an assemblage produced in a one-sided SDI steering scenario, then demonstrating superunsteerability is equivalent to steerability of the assemblage.
This form of steerability is due to the quantum advantage in using a quantum system of lower local Hilbert space dimension over the requirement of high dimensionality of the hidden variables.
Since steering is truly a quantum phenomenon \cite{Sch35}, superunsteerability captures  a genuinely quantum effect, although it does not indicate steerability in the strongest way as the WJD's form of steerability occurs for certain entangled states \cite{WJD07}. 
Just like the WJD's form of steerability can be witnessed through a suitable inequality or criterion, in Ref. \cite{JD23}, we have shown a certification of quantum discord which provides sufficient criterion for determining superunsteerability of the two-qubit discordant states.  
While discord has been operationally understood as having classical correlations assisted with quantum coherence  rather than quantum correlations \cite{Vedral2017}, our results in this work and our recent other works on superunsteerability \cite{DBD+18, JDK+19,JD23} thus bring the insight that a certain discordant states have also quantum correlations exhibiting quantum steering, just like a certain entangled states have steerability, though in a weaker form.

\textit{Acknowledgements.} CJ is grateful to Dr. Ashutosh Rai, Dr. Shin-Liang Chen, Dr. Manik Banik and Prof. Wei-Min Zhang for helpful discussions
and would like to thank the National Science
and Technology Council (formerly Ministry of Science and
Technology), Taiwan  (Grant No. MOST 111-2811-M-006-040-MY2). DD acknowledges the Royal Society (United Kingdom) for the support through the Newton International Fellowship (NIF$\backslash$R$1\backslash212007$). RS acknowledges support from Interdisciplinary Cyber Physical Systems (ICPS) programme of the Department of Science and Technology (DST), India, Grant No.: DST/ICPS/QuST/Theme-1/2019/6.

\textit{Appendix A: Nosignaling Boxes.}
Consider a scenario consisting of two  spatially separated observers, say Alice and Bob. Let Alice  has access to  inputs $x$ and observes outputs $a$ and Bob has access to inputs $y$ and observes outputs $b$. In this scenario,  a bipartite box $P$ = $P(ab|xy)$ := $\{ p(ab|xy) \}_{a,x,b,y}$ is the set of joint probability distributions $p(ab|xy)$ for all possible $a$, $x$, $b$, $y$. The single-partite   
box $P(a|x)$ := $\{p(a|x)\}_{a,x}$ of a NS box $P(ab|xy)$ 
is the set of marginal probability distributions $p(a|x)$ for all possible $a$ and $x$; which are given by,
\be
p(a|x)=\sum_b p(ab|xy), \quad \forall a,x,y.
\ee
The single-partite box $P(b|y)$ := $\{p(b|y)\}_{b,y}$ of a NS box $P(ab|xy)$ is the set of marginal
probability distributions 
$p(b|y)$ for all possible $b$ and $y$; which are given by,
\be
p(b|y)=\sum_a p(ab|xy), \quad \forall b,x,y.
\ee

A NS box  $P(ab|xy)$ is  local if and only if it can be reproduced by a LHV model, 
\begin{equation}
p(ab|xy)=\sum_\lambda p(\lambda) p(a|x,\lambda)p(b|y,\lambda) \hspace{0.3cm} \forall a,b,x,y;
\end{equation}
where $\lambda$ denotes shared randomness which occurs with probability $p(\lambda)$; each $p(a|x,\lambda)$ and $p(b|y,\lambda)$ are conditional probabilities. Otherwise, it is nonlocal.
The set of local boxes which have a LHV model forms a convex polytope called local 
polytope.

In this work, we have considered  a steering scenario where each of Alice and Bob performs two local measurements and each measurement has two possible outcomes. The set of all correlations produced in this scenario forms a subset of the set of NS bipartite correlations with two binary inputs and binary outputs for Alice and Bob, i.e., $x, y, a, b \in \{0, 1\}$.  In this case, the local polytope  has $16$ extremal boxes 
which are local-deterministic boxes given by,
\be \label{LDB}
    P_{D}^{\alpha \beta \gamma \epsilon} (ab|xy) = 
\begin{dcases}
    1,& \text{if } a = \alpha x \oplus \beta, b = \gamma y \oplus \epsilon \\
    0,              & \text{otherwise}.
\end{dcases}
\ee
Here, $\alpha, \beta, \gamma, \epsilon \in \{0,1\}$ and $\oplus$ denotes addition modulo $2$. Any local box can be written
as a convex mixture of  the above extremal local-deterministic boxes.  All the extremal local-deterministic boxes  as defined above
can be written as the product of marginals corresponding to Alice and Bob,
i.e., $P_D^{\alpha\beta\gamma\epsilon}(ab|xy)=P^{\alpha\beta}_D(a|x)P^{\gamma\epsilon}_D(b|y)$,
with the  deterministic box on Alice's side given by
\begin{equation}
P_D^{\alpha\beta}(a|x)=\left\{
\begin{array}{lr}
1, & a=\alpha x\oplus \beta\\
0 , & \text{otherwise}\\
\end{array}
\right. 
\label{}
\end{equation} 
and the  deterministic box on Bob's side given by,
\begin{equation}
P_D^{\gamma\epsilon}(b|y)=\left\{
\begin{array}{lr}
1, & b=\gamma x\oplus \epsilon\\
0 , & \text{otherwise}.\\
\end{array}
\right.   
\label{}
\end{equation}

A local box satisfies the complete set of Bell  inequalities \cite{WW01}. In the case 
of  two-binary-inputs and two-binary-outputs, the Bell-CHSH inequalities \cite{CHS+69} are given by,
\begin{eqnarray}
\label{chsh}
\mathcal{B}_{\alpha \beta \gamma} =&& (-1)^{\gamma} \langle A_0 B_0 \rangle + (-1)^{\beta \oplus \gamma} \langle A_0 B_1 \rangle \nonumber\\
&&+ (-1)^{\alpha \oplus \gamma} \langle A_1 B_0 \rangle + (-1)^{\alpha \oplus \beta \oplus \gamma \oplus 1} \langle A_1 B_1 \rangle \leq 2,
\end{eqnarray}
where $\alpha, \beta, \gamma \in  \{0, 1 \}$,   $\langle A_x B_y \rangle =  \sum_{a,b} (-1)^{a\oplus b} p(ab|x y)$.
The above set of Bell-CHSH inequalities form the complete set of Bell inequalities.
All these tight Bell inequalities form the nontrivial facets for the local polytope. 
All nonlocal boxes lie outside the local polytope and violate a Bell inequality.

The set of all bipartite  two-input-two-output NS boxes forms an $8$ dimensional convex polytope with $24$ extremal boxes 
\cite{BLM+05}, which can be divided into two classes:  i) nonlocal boxes having $8$ Popescu-Rohrlich (PR) boxes which are equivalent to the canonical PR box \cite{PR94} under ``local reversible operations" (LRO)
and  ii) local boxes having $16$ local-deterministic boxes, which are given in Eq. (\ref{LDB}). By using LRO Alice and Bob can convert any extremal box in one class into any other extremal box within the same class. LRO is designed \cite{BLM+05} as follows: Alice may relabel her inputs: $x \rightarrow x \oplus 1$, and she may relabel her outputs (conditionally on the input) : $a \rightarrow a \oplus \alpha x \oplus \beta$; Bob can perform similar operations.

\textit{Appendix B: Proof of Lemma $1$.}

Here we provide the proof of Lemma $1$ in the main text that the LHV-LHS decomposition of the following correlation:
 \begin{equation}      
 \label{corr1}
 P(ab|xy) =\begin{tabular}{c|cccc}
 \backslashbox{(x,y)}{(a,b)} & (0,0) & (0,1) & (1,0) & (1,1)\\\hline\\[0.05cm]
(0,0) & $\dfrac{1}{2}$ & $\dfrac{1}{4}$ & $0$ & $\dfrac{1}{4}$ \\[0.5cm]
(0,1) & $\dfrac{3}{8}$ & $\dfrac{3}{8}$ & $\dfrac{1}{8}$ & $\dfrac{1}{8}$ \\[0.5cm]
(1,0) & $\dfrac{1}{4}$ & $\dfrac{1}{2}$ & $\dfrac{1}{4}$ & $0$ \\[0.5cm]
(1,1) & $\dfrac{3}{8}$ & $\dfrac{3}{8}$ & $\dfrac{1}{8}$ & $\dfrac{1}{8}$ \\
\end{tabular}
\end{equation}
 cannot be realized with hidden variable having dimension $2$.

At first, let us try to generate a LHV-LHS decomposition of the correlation given by Eq.(\ref{corr1})  with hidden variables of dimension $2$ having different deterministic distributions at Alice's side. In this case the correlation (\ref{corr1}) can be decomposed in the following way:
\begin{equation}
P(ab|x y) = \sum_{\lambda=0}^{1} p(\lambda) P(a|x,\lambda) P(b |y,  \rho_{\lambda}),
\end{equation}
where $0 < p(\lambda) < 1$ for all $\lambda \in \{0,1\}$, $p(0)+p(1) =1$. Since Alice's strategies are deterministic, the two probability distributions $P(a|x,0)$ and $P(a|x,1)$ are different and they must be equal to any two among $P_D^{00}(a|x)$, $P_D^{01}(a|x)$, $P_D^{10}(a|x)$ and $P_D^{11}(a|x)$. It can be easily checked that none of these possible choices will satisfy all the joint probability distributions of the box (\ref{corr1}) considering arbitrary distributions at Bob's end with or without quantum realizations. It is, therefore, impossible to generate a LHV-LHS decomposition of the correlation given by Eq.(\ref{corr1}) with hidden variables of dimension $2$ having different deterministic distributions at Alice's side. 

Now, let us try to generate a LHV-LHS decomposition of the above correlation with hidden variables having dimension $2$ and with different \textit{non-deterministic} distributions at Alice's side. This  can be achieved by constructing a LHV-LHS model with hidden variables of dimension $4$ or $3$ having different deterministic distributions at Alice's side followed by taking equal joint probability distributions (having quantum realisations) at Bob's side as common and making the corresponding distributions at Alice's side non-deterministic.

If the hidden variable dimension in the LHV-LHS decomposition  of the above correlation can be reduced from $4$ to $2$, then the above correlation can be decomposed in the following way,
\begin{equation}
\label{new11}
P(ab|x y) = \sum_{\lambda=0}^{3} p(\lambda) P(a|x,\lambda) P(b |y,  \rho_{\lambda}),
\end{equation}
where $P(a|x,\lambda)$ are different deterministic distributions (without any loss of generality we can assume that $P(a|x,0) = P_D^{00}(a|x)$, $P(a|x,1) = P_D^{01}(a|x)$, $P(a|x,2) = P_D^{10}(a|x)$, $P(a|x,3) = P_D^{11}(a|x)$);
and either any three of the four distributions $P(b  |y, \rho^{\lambda})$ 
are equal to each other or there exists two sets each containing two equal distributions $P(b  |y, \rho^{\lambda})$; $0 < p(\lambda) < 1$ for $\lambda$ = $0,1,2,3$; 
$\sum_{\lambda=0}^{3} p(\lambda) = 1$. Then, as described earlier, taking equal probability distributions 
$P(b |y, \rho_{\lambda})$ at Bob's end as common and making corresponding distribution at Alice's side non-deterministic will reduce the dimension of the hidden variable from $4$ to $2$.

There are the following seven cases in which the dimension of the hidden variable in the LHV-LHS decomposition of the above correlation can be reduced from $4$ to $2$: 
\begin{align}
	& (a) \, \,  P(b |y, \rho_{0}) = P(b |y, \rho_{1}) = P(b |y, \rho_{2}); \nonumber \\ 
	& (b)  \, \, P(b |y, \rho_{0}) = P(b |y, \rho_{1}) = P(b |y, \rho_{3}); \nonumber \\ 
	& (c)  \, \, P(b |y, \rho_{0}) = P(b |y, \rho_{2}) = P(b |y, \rho_{3}); \nonumber \\ 
	& (d)  \, \, P(b |y, \rho_{1}) = P(b |y, \rho_{2}) = P(b |y, \rho_{3}); \nonumber \\
	& (e)  \, \, P(b |y, \rho_{0}) = P(b |y, \rho_{1}) \, \, \text{as well as} \, \,  P(b |y, \rho_{2}) = P(b |y, \rho_{3}); \nonumber\\
	& (f)  \, \, P(b |y, \rho_{0}) = P(b |y, \rho_{2}) \, \, \text{as well as} \, \, P(b |y, \rho_{1}) = P(b |y, \rho_{3}); \nonumber\\
	& (g)  \, \,  P(b |y, \rho_{0}) = P(b |y, \rho_{3}) \, \, \text{as well as} \, \, P(b |y, \rho_{1}) = P(b |y, \rho_{2}); \nonumber
\end{align}
Now in any of these possible cases, considering arbitrary probability distributions 
$P(b |y, \rho_{\lambda})$ at Bob's end (without considering any constraint), it can be shown that all the probability distributions of (\ref{corr1}) are not reproduced simultaneously. For example, consider the case (a) mentioned above. In this case, if we want to satisfy $P(10|00) = 0$, then $P(00|00)$ also becomes $0$, which is not true for (\ref{corr1}). Similarly, for the case (b) mentioned above, if we want to satisfy $P(10|00) = 0$, then $P(00|00)$ also becomes $0$. In a similar way, one can show that all probability distributions of (\ref{corr1}) are not reproduced simultaneously for the other five cases.
Hence, this also holds when the boxes $P(b |y, \rho_{\lambda})$ have quantum realizations. Therefore, the correlation (\ref{corr1}) cannot have a LHV-LHS decomposition with hidden variables having dimension $4$ such that this dimension can be reduced to $2$.

Next, let us check whether the correlation (\ref{corr1}) can have a LHV-LHS decomposition with hidden variables having dimension $3$ such that this dimension can be reduced to $2$. In this case, the above correlation can be decomposed as
\begin{equation}
\label{new11}
P(ab|x y) = \sum_{\lambda=0}^{2} p(\lambda) P(a|x,\lambda) P(b |y,  \rho_{\lambda}),
\end{equation}
where $P(a|x,\lambda)$ are different deterministic distributions;
and any two of the four distributions $P(b  |y, \rho^{\lambda})$ 
are equal to each other; $0 < p(\lambda) < 1$ for $\lambda$ = $0,1,2$; 
$\sum_{\lambda=0}^{2} p(\lambda) = 1$. Then, as described earlier, taking equal joint probability distributions 
$P(b |y, \rho_{\lambda})$ at Bob's end as common and making corresponding distribution at Alice's side non-deterministic will reduce the dimension of the hidden variable from $3$ to $2$.

There are the following twelve cases in which the dimension of the hidden variable in the LHV-LHS decomposition of the above correlation can be reduced from $3$ to $2$: 
\begin{align}
	& (a) \, \,  P(b |y, \rho_{0}) = P(b |y, \rho_{1}) \nonumber \\
 & \hspace{1cm} \, \, \text{with} \, \, P(a|x,0) = P_D^{00}, P(a|x,1) = P_D^{01}, P(a|x,2) = P_D^{10}; \nonumber \\ 
	& (b) \, \,  P(b |y, \rho_{1}) = P(b |y, \rho_{2}) \nonumber \\
 & \hspace{1cm} \, \, \text{with} \, \, P(a|x,0) = P_D^{00}, P(a|x,1) = P_D^{01}, P(a|x,2) = P_D^{10}; \nonumber \\
 & (c) \, \,  P(b |y, \rho_{0}) = P(b |y, \rho_{2}) \nonumber \\
 & \hspace{1cm} \, \, \text{with} \, \, P(a|x,0) = P_D^{00}, P(a|x,1) = P_D^{01}, P(a|x,2) = P_D^{10}; \nonumber \\
 & (d) \, \,  P(b |y, \rho_{0}) = P(b |y, \rho_{1}) \nonumber \\
 & \hspace{1cm} \, \, \text{with} \, \, P(a|x,0) = P_D^{00}, P(a|x,1) = P_D^{01}, P(a|x,2) = P_D^{11}; \nonumber \\ 
& (e) \, \,  P(b |y, \rho_{1}) = P(b |y, \rho_{2}) \nonumber \\
 & \hspace{1cm} \, \, \text{with} \, \, P(a|x,0) = P_D^{00}, P(a|x,1) = P_D^{01}, P(a|x,2) = P_D^{11}; \nonumber \\
 & (f) \, \,  P(b |y, \rho_{0}) = P(b |y, \rho_{2}) \nonumber \\
 & \hspace{1cm} \, \, \text{with} \, \, P(a|x,0) = P_D^{00}, P(a|x,1) = P_D^{01}, P(a|x,2) = P_D^{11}; \nonumber \\
 & (g) \, \,  P(b |y, \rho_{0}) = P(b |y, \rho_{1}) \nonumber \\
 & \hspace{1cm} \, \, \text{with} \, \, P(a|x,0) = P_D^{00}, P(a|x,1) = P_D^{10}, P(a|x,2) = P_D^{11}; \nonumber \\ 
& (h) \, \,  P(b |y, \rho_{1}) = P(b |y, \rho_{2}) \nonumber \\
 & \hspace{1cm} \, \, \text{with} \, \, P(a|x,0) = P_D^{00}, P(a|x,1) = P_D^{10}, P(a|x,2) = P_D^{11}; \nonumber \\
 & (i) \, \,  P(b |y, \rho_{0}) = P(b |y, \rho_{2}) \nonumber \\
 & \hspace{1cm} \, \, \text{with} \, \, P(a|x,0) = P_D^{00}, P(a|x,1) = P_D^{10}, P(a|x,2) = P_D^{11}; \nonumber \\
  & (j) \, \,  P(b |y, \rho_{0}) = P(b |y, \rho_{1}) \nonumber \\
 & \hspace{1cm} \, \, \text{with} \, \, P(a|x,0) = P_D^{01}, P(a|x,1) = P_D^{10}, P(a|x,2) = P_D^{11}; \nonumber \\ 
& (k) \, \,  P(b |y, \rho_{1}) = P(b |y, \rho_{2}) \nonumber \\
 & \hspace{1cm} \, \, \text{with} \, \, P(a|x,0) = P_D^{01}, P(a|x,1) = P_D^{10}, P(a|x,2) = P_D^{11}; \nonumber \\
 & (l) \, \,  P(b |y, \rho_{0}) = P(b |y, \rho_{2}) \nonumber \\
 & \hspace{1cm} \, \, \text{with} \, \, P(a|x,0) = P_D^{01}, P(a|x,1) = P_D^{10}, P(a|x,2) = P_D^{11}; \nonumber 
\end{align}
For any of these above-mentioned cases, considering arbitrary distributions 
$P(b |y, \rho_{\lambda})$  (without imposing any constraint), it can be shown that all the probability distributions of (\ref{corr1}) are not reproduced simultaneously. For example, consider the case (a) mentioned above. In this case, if we want to satisfy $P(10|00) = 0$, then $P(00|10)$ also becomes $0$, which is not true for (\ref{corr1}). Similarly, for the case (b) mentioned above, if we want to satisfy $P(11|10) = 0$, then $P(11|00)$ also becomes $0$ -- incompatible with (\ref{corr1}). In a similar way, one can show that all probability distributions of (\ref{corr1}) are not reproduced simultaneously for the other ten cases.
Hence, this also holds when each of the boxes $P(b |y, \rho_{\lambda})$ has a quantum realization. Therefore, the correlation (\ref{corr1}) cannot have a LHV-LHS decomposition with hidden variables having dimension $3$ such that this dimension can be reduced to $2$. This completes the proof.




\bibliography{oneway}

\end{document}